\documentclass[12pt]{article}

\usepackage{hyperref}
\usepackage{amsfonts, amsmath, amssymb, amsthm}
\usepackage{upgreek}
\usepackage{a4}

\usepackage{color}
 \definecolor{mdg}{RGB}{0,177,0} 
 \definecolor{mdb}{RGB}{0,0,191}
 \definecolor{mddb}{RGB}{0,0,91}
 \definecolor{mdy}{RGB}{255,69,0} 
 \definecolor{gray}{RGB}{99,99,99} 

    \definecolor{refkey}{RGB}{255,127,0}
    \definecolor{labelkey}{RGB}{127,0,255}

\newcommand{\boldeqref}[1]{\textbf{(\ref{#1})}} 

\DeclareMathOperator{\linearspan}{span}

\newtheorem{theorem}{Theorem}
\newtheorem{lemma}{Lemma}

\theoremstyle{definition}
\newtheorem{dfn}{Definition}
\newtheorem{cnd}{Condition}
\newtheorem{cnv}{Convention}
\newtheorem{mpt}{Assumption}

\theoremstyle{remark}
\newtheorem{remark}{Remark}

\author{I.G. Korepanov}
\title{Free fermions on a piecewise linear four-manifold. I:~Exotic chain complex}
\date{November 2016}

\begin{document}

\maketitle

\begin{abstract}
Recently, an algebraic realization of the four-dimensional Pachner move 3--3 was found in terms of Grassmann--Gaussian exponentials, and a remarkable nonlinear parameterization for it, going in terms of a $\mathbb C$-valued 2-cocycle. Here we define, for a given triangulated four-dimensional manifold and a 2-cocycle on it, an `exotic' chain complex intimately related to the mentioned parameterization, thus providing a basis for algebraic realizations of all four-dimensional Pachner moves.
\end{abstract}

\section{Introduction}\label{s:intro}

This paper is intended to provide a basis for constructing algebraic realizations of all four-dimensional Pachner moves. The fundamental role in our constructions is played by Grassmann--Gaussian exponentials, that is, exponentials of quadratic forms of anticommuting variables. Such exponentials are usually associated with free-fermionic quantum field theories (QFT), so our theory can be regarded as a topological QFT of this kind.

Eventually, we plan to obtain an invariant of a pair $(M,h)$, where $M$ is a closed\footnote{Recall that `closed' means in this context `compact and without boundary'.} oriented connected four-dimen\-sional piecewise linear (PL) manifold, and $h\in H^2(M,\mathbb C)$ is an element of the second cohomology group with complex numbers as coefficients.\footnote{Strictly speaking, we will need also a chosen basis in~$H^2(M,\mathbb C)$, but given only up to a linear transformation with determinant~$\pm 1$, see Remark~\ref{r:b} below.}

One obvious advantage of working with PL manifolds, as opposed to smooth manifolds, is that we typically deal with a \emph{finite} number of simplices in which the manifold is decomposed (triangulated), and attach also a finite number of quantities to these simplices. In such a \emph{discrete} theory, we will not encounter many difficulties such as appear when we try to define a functional integral on a smooth manifold.\footnote{Our analogue of functional integration will be a \emph{Berezin integral} in anticommuting (Grassmann) variables, it will appear in the next part~\cite{part-II} of this work. Note that, even in our discrete theory, some infinities that must be `renormalized' do occur. Namely, they appear in disguise, as zeros $0=\infty^{-1}$: we will get zero if we try to write, in a na\"ive way, a manifold invariant using the algebraic realization of Pachner move 3--3 from~\cite{full-nonlinear}. The `exotic complex' proposed in the present paper can be treated as a tool for such renormalization, this latter consisting in replacing a vanishing determinant with a Reidemeister-type \emph{torsion}.}

Here we present the first part of our construction, where we assume that $M$ is equipped with one fixed triangulation. Our aim is to build an unusual chain complex, using this triangulation and a 2-cocycle~$\omega$ representing class~$h$. Our invariant will appear in Part~II~\cite{part-II} of this work; it will be the square root of the Reidemeister torsion of this complex~\eqref{c-symb}, multiplied by some correcting factor, and this latter will be the product of same-type factors attached to all triangulation simplices in dimensions 2 and~4.

Of course, an invariant of the pair~$(M,h)$ deserves its name if it does \emph{not} depend on a triangulation of~$M$ or specific cocycle~$\omega \in h$. As we leave the actual definition of our invariant, as well as the proof of its invariance, until paper~\cite{part-II}, we must offer here a convincing argument why the construction presented in this Part~I has something to do with the invariance.

Recall that any triangulation of a manifold can be transformed into another triangulation using a finite sequence of elementary rebuildings called \emph{Pachner moves}~\cite{Pachner}; excellent pedagogical introduction in this theory can be found in the monograph~\cite{Lickorish}. In four dimensions, the `central' move, in some informal sense, is move 3--3: it transforms a cluster of three 4-simplices situated around a 2-face into a cluster of three other 4-simplices, occupying the same place in the triangulation. An explicit \emph{algebraic realization} of move 3--3 was given in papers~\cite{full-nonlinear,gce} using Grassmann and Clifford algebras. Here algebraic realization means a formula whose structure corresponds naturally to the move 3--3, in such way that it gives strong support for further construction of a kind of topological quantum field theory.

Our immediate argument in defense of chain complex~\eqref{c-symb} proposed in this paper is as follows: its construction can be shown to lead to the same, up to a `gauge' freedom (see Subsection~\ref{ss:expl}), \emph{pentachoron weights}~\eqref{W_u} as appear in~\cite{full-nonlinear} and~\cite{gce}\footnote{The \emph{very same} pentachoron weight as in~\cite{gce} appears in the first arXiv version \href{http://arxiv.org/abs/1605.06498v1}{arXiv:1605.06498v1} of this paper. Here we prefer, however, to use another gauge.}. We will actually show in~\cite{part-II} that complex~\eqref{c-symb} provides an invariant of \emph{all} Pachner moves, while papers~\cite{full-nonlinear} and~\cite{gce} deal only with the move 3--3.

\begin{remark}
The existence of our algebraic realization of move 3--3 was first shown, and in a constructive way, in~\cite{KS2}. What was lacking in~\cite{KS2} was the natural \emph{nonlinear parameterization in terms of a 2-cocycle}, found later in~\cite{full-nonlinear}. Once this parameterization is unveiled, further work leads --- again quite naturally --- to the constructions presented in this paper.
\end{remark}

As we have said, we will be dealing here with one chosen 2-cocycle~$\omega$ representing class~$h$. This means that a complex number $\omega_s=\omega_{ijk}$ is given for any oriented triangle~$s$ in our fixed triangulation, with vertices $i,j,k\in s$, such that for any tetrahedron $t=ijkl$
\begin{equation}\label{c}
\omega_{jkl}-\omega_{ikl}+\omega_{ijl}-\omega_{ijk}=0.
\end{equation}
Similarly, all (co)chains in this paper are understood as simplicial (co)chains for our fixed triangulation (although sometimes with peculiar modifications, see Section~\ref{s:t}).

Here and below we use the following notational conventions.

\begin{cnv}\label{cnv:v}
All the vertices (0-simplices) in the triangulation are assumed to be numbered from~$1$ to their total number~$N_0$. Vertices are also denoted by letters $i,j,k,\dots$.
\end{cnv}

\begin{cnv}\label{cnv:n}
We denote triangles (2-simplices) by the letter~$s$, tetrahedra (3-simplices) by~$t$, and pentachora (4-simplices) --- by~$u$. As for edges (1-simplices), we tend to use letter~$b$ for them. If needed, we also number $d$-dimensional simplices from~$1$ to their total number~$N_d$.
\end{cnv}

\begin{cnv}\label{cnv:o}
We also write simplices by their vertices, e.g., $s=ijk$ or $u=12345$. In writing so, we assume that the simplex has the \emph{orientation} determined by the vertex order. Also, minus sign (as in `$-t$', etc.) means reversing the orientation; in this sense, $ijk=-jik$, etc.
\end{cnv}

Also, we adopt the following convention about cocycle~$\omega$.

\begin{cnv}\label{cnv:g}
Cocycle~$\omega$ is \emph{generic} in the following sense: statements and equalities in this paper involving~$\omega$ (like ``rank of a certain matrix is~5'') actually hold for almost all (${}={}$Zariski open set of) cocycles. The denominators in such formulas as~\eqref{Ftt'} don't vanish for a generic~$\omega$ either.
\end{cnv}

Below,
\begin{itemize}\itemsep 0pt
 \item in Section~\ref{s:GB}, we recall some facts from the Grass\-mann--Berezin calculus of anticommuting variables that we will be using;
 \item in Section~\ref{s:t}, we begin the definition of our exotic complex: we define its vector spaces, their bases, and a (simply defined) morphism~$f_1$. Due to the symmetry of the complex, there remain only two morphisms to be defined, called $f_2$ and~$f_3$;
 \item in Section~\ref{s:f2}, we define~$f_2$. It leads, together with its `companion mapping'~$g_2$, to an important notion of `edge operators';
 \item in Section~\ref{s:long}, we do more work on~$f_2$, leading to explicit formulas for its matrix elements;
 \item in Section~\ref{s:f3}, we define morphism~$f_3$, prove the chain property for our exotic complex, and calculate matrix elements of~$f_3$;
 \item finally, in Section~\ref{s:d}, we discuss what has been done in this paper, and plans for its parts II and~III.
\end{itemize}

\section{Some simple facts from Grassmann--Berezin calculus of anticommuting variables}\label{s:GB}

\subsection{Grassmann algebras}\label{ss:GB}

\begin{dfn}\label{dfn:GA}
In this paper, \emph{Grassmann algebra} will mean a finite-dimen\-sional associative algebra over the field~$\mathbb C$ of complex numbers, with unity, generators~$\vartheta _i$, and relations
\begin{equation}\label{xixj}
\vartheta _i \vartheta _j = -\vartheta _j \vartheta _i .
\end{equation}
\end{dfn}

Generators~$\vartheta _i$ will also be called \emph{Grassmann variables}.\footnote{We thus switch from the notation~$x_i$ for Grassmann variables, used in our previous works and earlier by Berezin~\cite{B} as well as Chevalley~\cite{C}, to $\vartheta _i$ that seems to be more common now.}

As~\eqref{xixj} implies that $\vartheta _i^2 =0$, each element of a Grassmann algebra is a polynomial of degree $\le 1$ in each~$\vartheta _i$.

The \emph{degree} of a Grassmann monomial is its total degree in all Grassmann variables. An element of Grassmann algebra consisting of monomials of only odd or only even degrees is called \emph{odd} or \emph{even}, respectively. If all monomials have degree~2, such element is called a \emph{quadratic form}.

The \emph{exponential} is defined by its usual Taylor series.

\begin{dfn}\label{dfn:GGe}
A \emph{Grassmann--Gaussian exponential} is the exponential of a quadratic form.
\end{dfn}

\begin{dfn}\label{dfn:der}
The (left) \emph{derivative} --- we also call it \emph{differentiation} --- with respect to a Grassmann variable~$\vartheta _i$, denoted $\dfrac{\partial}{\partial \vartheta _i}$ or simply~$\partial_i$, is a $\mathbb C$-linear operation in Grassmann algebra, satisfying also the following conditions:
\[
\partial_i 1=0, \qquad \partial_i \vartheta _i =1, \qquad \partial_i \vartheta _j =0 \text{ \ for \ }i\ne j,
\]
and the \emph{Leibniz rule}: for $f$ either even or odd,
\begin{equation}\label{Leibniz}
\partial_i (fg) = (\partial_i f) g + \epsilon f\partial_i g,
\end{equation}
where $\epsilon=1$ for an even~$f$ and $\epsilon=-1$ for an odd~$f$.
\end{dfn}

\subsection{Clifford algebra generated by differentiations w.r.t.\ and multiplications by Grassmann variables}\label{ss:C}

Let there be a Grassmann algebra with~$n$ generators~$\vartheta _i$; we denote $U=\linearspan_{\mathbb C}\{\vartheta _i\}$ the linear space spanned by them. As $\partial_i\vartheta _j$ is always a number, differentiations are identified with \emph{linear forms} on~$U$, hence the dual space $U^*=\linearspan_{\mathbb C}\{\partial_i\}$. Below we also identify elements of~$U$ with ($\mathbb C$-) linear operators --- left multiplications by these elements in our Grassmann algebra.

Of fundamental importance for us will be the space $\mathcal V=U\oplus U^*$, consisting thus of linear operators
\begin{equation}\label{bg}
d = \sum_{t=1}^n (\beta_t\partial_t + \gamma_t\vartheta _t), \qquad \beta_t,\gamma_t\in\mathbb C.
\end{equation}

The anticommutator of two operators~\eqref{bg} (defined as $[A,B]_+=AB+BA$ for operators $A$ and~$B$) has only scalar part; we call this their \emph{scalar product}:
\begin{equation}\label{scal}
\langle d^{(1)} \mid d^{(2)}\rangle \stackrel{\rm def}{=} \sum_{t=1}^n (\beta_t^{(1)}\gamma_t^{(2)} + \beta_t^{(2)}\gamma_t^{(1)} ).
\end{equation}
With this scalar product, $\mathcal V$ becomes a \emph{complex Euclidean space}, while all polynomials of operators~\eqref{bg} form a \emph{Clifford algebra}.

Separate summands in \eqref{bg} and~\eqref{scal} will also be of use for us, so we introduce \emph{$t$-components}
\begin{equation}\label{scs}
d|_t \stackrel{\rm def}{=} \beta_t\partial_t+\gamma_t \vartheta _t
\end{equation}
of operator~$d$, and \emph{partial scalar products}
\begin{equation}\label{sct}
\langle d^{(1)}\mid d^{(2)}\rangle_t \stackrel{\rm def}{=} \beta_t^{(1)}\gamma_t^{(2)} + \beta_t^{(2)}\gamma_t^{(1)}.
\end{equation}
Also, we denote~$\mathcal V_t$ the two-dimensional linear space spanned by $\partial_t$ and~$\vartheta _t$; it is clear that $\mathcal V = \bigoplus_{t=1}^n \mathcal V_t$ in the sense of complex Euclidean spaces.

\subsection{Maximal isotropic spaces of operators}\label{ss:max-iso}

\begin{dfn}\label{dfn:iso}
Subspace~$V$ of a complex Euclidean space is called \emph{isotropic} if the scalar product restricted onto~$V$ identically vanishes.
\end{dfn}

Especially interesting for us are maximal --- in the sense of inclusion --- isotropic subspaces. Two simplest examples of maximal isotropic subspaces are $U\subset \mathcal V$ and $U^*\subset \mathcal V$ in the previous Subsection~\ref{ss:C}.

The following Theorems \ref{th:eo} and~\ref{th:v} summarize some well-known facts about maximal isotropic subspaces, skew-symmetric matrices, and Grassmann--Gaussian exponentials. These (and many other interesting) facts were known already to Chevalley~\cite{C}; later essentially the same mathematical structures were studied by Ab{\l}amowicz and Lounesto~\cite{AL} and Fauser~\cite{Fauser}. Also, Theorems \ref{th:eo} and~\ref{th:v} can be regarded as light-weight versions of~\cite[Theorems 1 and~2]{full-nonlinear}, where the reader can find simple proofs.

\begin{theorem}\label{th:eo}
Let $\vartheta _t$, \ $t=1,\ldots,n$, be generators of a Grassmann algebra over~$\mathbb C$, and $\mathcal V$ the $2n$-dimensional complex Euclidean space of operators~\eqref{bg}. Let also $\mathsf p=\begin{pmatrix}\partial_1 & \dots & \partial_n\end{pmatrix}^{\mathrm T}$ and\/ $\uptheta =\begin{pmatrix}\vartheta _1 & \dots & \vartheta _n\end{pmatrix}^{\mathrm T}$ be the columns of partial derivatives and corresponding variables. Then,
\begin{enumerate}\itemsep 0pt
\item \label{i:l} for a skew-symmetric matrix~$F$, the span of the elements of column $\mathsf p + F \uptheta $ is a maximal isotropic subspace,
\item \label{i:f} conversely, if there is an isotropic subspace spanned by operators of the form
\begin{equation}\label{pFx}
 \begin{array}{c}
\partial_1+F_{11}\vartheta _1+\dots+F_{1n}\vartheta _n,\\
\dotfill \\
\partial_n+F_{n1}\vartheta _1+\dots+F_{nn}\vartheta _n,
 \end{array}
\end{equation}
then it is maximal, and numbers~$F_{ij}$ form a skew-symmetric matrix.
\end{enumerate}
\qed
\end{theorem}

\begin{theorem}\label{th:v}
Let $\vartheta _t$, \ $t=1,\ldots,n$, be generators of a Grassmann algebra over~$\mathbb C$, and $V \subset \mathcal V$ a maximal isotropic subspace in the space of operators~\eqref{bg}. Then,
\begin{enumerate}\itemsep 0pt
 \item\label{i:t} the nullspace of~$V$ (${}={}$the space of vectors annihilated by all elements in~$V$) is one-dimensional,
 \item\label{i:u} if $V$ is a subspace of the kind dealt with in items \ref{i:l} and~\ref{i:f} of Theorem~\ref{th:eo}, i.e.,
  \begin{equation}\label{Vs}
V = \bigl(\text{span of the elements of column }(\mathsf p + F \uptheta )\bigr),
  \end{equation}
then the nullspace of\/~$V$ is spanned by the Grassmann--Gaussian exponential
  \begin{equation}\label{xFx}
  \mathcal W = \exp \left(-\frac{1}{2}\,\uptheta ^{\mathrm T} F \uptheta  \right).
  \end{equation}
\end{enumerate}
\qed
\end{theorem}

\section{Exotic complex: the simpler part of its definition}\label{s:t}

Our ``exotic'' chain complex will be made of based --- that is, equipped with distinguished bases --- $\mathbb C$-linear spaces and their morphisms, identified naturally with matrices. It will have the following form:
\begin{multline}\label{c-symb}
0 \longrightarrow \mathbb C \stackrel{f_1}{\longrightarrow} Z^2(M,\mathbb C) \stackrel{f_2}{\longrightarrow} \tilde C_3(M,\mathbb C) \stackrel{f_3=-f_3^{\mathrm T}}{\longrightarrow} \tilde C^3(M,\mathbb C) \\
\stackrel{f_4=f_2^{\mathrm T}}{\longrightarrow} \bigl(Z^2(M,\mathbb C)\bigr)^* \stackrel{f_5=f_1^{\mathrm T}}{\longrightarrow} \mathbb C \longrightarrow 0.
\end{multline}
The right half of this complex is, in a sense, a mirror image of its left half: the fourth, fifth and sixth spaces are, by definition, \emph{dual} to the third, second and first space, respectively, and morphisms between these dual spaces are given by the transposed matrices of ``initial'' morphisms, as is indicated above the arrows in~\eqref{c-symb}. Note that matrix~$f_3$ --- central not only in complex~\eqref{c-symb} but in our whole theory --- is \emph{skew-symmetric}.

In this Section, we present the description of the spaces taking part in this complex, their bases, and morphism~$f_1$. Morphisms $f_2$ and~$f_3$ are more complicated, and their description is developed in Sections \ref{s:f2}--\ref{s:f3} below.

\paragraph{Linear spaces in complex~\boldeqref{c-symb}.}
The second (nonzero) space $Z^2(M,\mathbb C)$ consists of all 2-cocycles on~$M$. The third space needs more explanation: it is a $\mathbb C$-linear space with the set of all \emph{unoriented} tetrahedra (in our fixed triangulation of~$M$) making its basis. So, the third space consists of ``unoriented 3-chains'', which is marked by a tilde above the usual letter~$C$.

\paragraph{Bases of linear spaces in complex~\boldeqref{c-symb}.}
We will need bases in our vector spaces in order to define a Reidemeister-type \emph{torsion}\footnote{This will be done in the next part of this work.}. This allows us to define the bases up to the following equivalence: two bases belong, by definition, to the same equivalence class, provided the determinant of their transition matrix is $\pm 1$. Also, the bases in spaces in the right half of~\eqref{c-symb} are, by definition, \emph{dual} to the bases in the first three spaces. So, we now define these three bases.

The first space~$\mathbb C$ in~\eqref{c-symb} is identified with the one-dimen\-sional space of 2-cocycles spanned by our given cocycle~$\omega$, and one vector~$\omega$ forms its basis. Note that $\omega$ is not identical zero, due to Convention~\ref{cnv:g}.

Basis in the second space~$Z^2(M,\mathbb C)$ is composed, by definition, from a basis in the space~$B^2(M,\mathbb C)$ of \emph{coboundaries}, and a pullback of some chosen basis in~$H^2(M,\mathbb C)$.

\begin{remark}\label{r:b}
We need a basis in~$H^2(M,\mathbb C)$, for our construction, of course only up to the same equivalence as stated above. If $H^2(M,\mathbb C)=0$, then empty basis works well, but otherwise, a class of bases must be chosen from some considerations.
\end{remark}

As for constructing the basis in~$B^2(M,\mathbb C)$, we will restrict ourself, in this paper, to the case formulated in the following Assumption.

\begin{mpt}\label{mpt:h}
Manifold~$M$ is \emph{simply connected}.
\end{mpt}

Assumption~\ref{mpt:h} implies that $H^1(M,\mathbb C)=0$, that is,
\begin{equation}\label{in}
B^2(M,\mathbb C)=C^1(M,\mathbb C)/Z^1(M,\mathbb C)=C^1(M,\mathbb C)/B^1(M,\mathbb C).
\end{equation}
Taking some liberty, we denote in this paper a 1-cochain taking value~$1$ on an edge~$b$ and vanishing on all other edges, simply by the same letter~$b$. Below, we denote by~$\delta$ the coboundary operator acting on $\mathbb C$-cochains on our triangulation.

\begin{remark}\label{r:d}
As is common in algebraic topology, we use the same letter~$\delta$ for coboundaries of cochains of all dimensions; this will not lead to a confusion.
\end{remark}

By definition, basis in~$B^2(M,\mathbb C)$ consists of elements~$\delta b$ for all $b\in \mathsf B$, where $\mathsf B$ is such a subset in the set of all edges that its complement~$\overline{\mathsf B}$ consists of all edges in a maximal tree in the 1-skeleton of our triangulation. Recall that 1-skeleton of a triangulation is the graph made of all its vertices and edges. 

\begin{remark}
Recalling also notation~$N_d$ from our Convention~\ref{cnv:n}, we can see that $\overline{\mathsf B}$ is a set of $N_0-1$ edges joining all the triangulation vertices. Hence, $\mathsf B$ contains the remaining $N_1-N_0+1$ edges.
\end{remark}

\begin{lemma}\label{l:edges}
\begin{enumerate}\itemsep 0pt
 \item\label{i:exact} The following sequence is exact:
\begin{equation}\label{io}
0\to \mathbb C\stackrel{\iota}{\to} C^0(M,\mathbb C)\stackrel{\delta}{\to} C^1(M,\mathbb C)\stackrel{\delta}{\to} B^2(M,\mathbb C)\to 0.
\end{equation}
Here $\iota$ sends a number $z\in\mathbb C$ to the cochain taking value~$z$ on every vertex in the triangulation (and $\delta$ was defined right above Remark~\ref{r:d}, see also this Remark itself).
 \item\label{i:B} Elements~$\delta b$ for $b\in \mathsf B$ constitute a basis in~$B^2(M,\mathbb C)$.
 \item\label{i:M} For a different choice of\/~$\mathsf B$ (but such that~$\overline{\mathsf B}$ still consists of all edges of a maximal tree), the basis belongs to the same class.
\end{enumerate}
\end{lemma}

\begin{proof}
\begin{itemize}\itemsep 0pt
 \item[\ref{i:exact}] Sequence~\ref{io} is exact due to the connectedness of~$M$ and~\eqref{in}.
 \item[\ref{i:B}] Consider submatrix~$Q$ of the central~$\delta$ in~\eqref{io}, corresponding to a subset of edges forming maximal tree~$\overline{\mathsf B}$. Clearly, $Q$ has the maximal rank, namely~$N_0-1$ (see notational Convention~\ref{cnv:v}). As the exact sequence~\eqref{io} must have a finite Reidemeister torsion, submatrix~$R$ of the right~$\delta$, corresponding to the remaining set~$\mathsf B$ of edges, must be (square and) nondegenerate (compare~\cite[Subsection~2.1]{Turaev}), and this proves this item.
 \item[\ref{i:M}] The determinant of~$Q$ minus any one column is clearly~$\pm 1$, for any~$\mathsf B$. This implies that $\det R$ is also always the same, up to at most a sign (compare again~\cite[Subsection~2.1]{Turaev}).
\end{itemize}
\end{proof}

\paragraph{Morphism~$\boldsymbol{f_1}$.}
Mapping~$f_1$ is, by definition, the obvious inclusion. In matrix form, it is simply the column consisting of the coordinates of our cocycle~$\omega$ in the basis of space~$Z^2(M,\mathbb C)$ described above (the description includes item~\ref{i:B} in Lemma~\ref{l:edges}). This way,
\begin{itemize}\itemsep 0pt
 \item $\dim B^2(M,\mathbb C)$ entries in the column~$f_1$ correspond to \emph{edges} $b\in\mathsf B$; we denote such elements~$\nu_b$,
 \item the remaining $\dim H^2(M,\mathbb C)$ elements correspond to cocycles $z\in Z^2(M,\mathbb C)$ whose classes form a basis~$\mathsf b$ in~$H^2(M,\mathbb C)$; we denote such elements~$\nu_z$,
 \item that the values~$\nu_b$ and~$\nu_z$ are the coordinates of~$\omega$ means that
\[
\sum_{b\in \mathsf B} \nu_b\, \delta b + \sum_{z\in \mathsf b} \nu_z\, z = \omega.
\]
\end{itemize}

\section[Morphism~$f_2$]{Morphism~$\boldsymbol{f_2}$}\label{s:f2}

\subsection[Definition of~$f_2$ and its companion mapping~$g_2$]{Definition of~$\boldsymbol{f_2}$ and its companion mapping~$\boldsymbol{g_2}$}\label{ss:d2}

To define morphism~$f_2\colon\; Z^2(M,\mathbb C)\to \tilde C_3(M,\mathbb C)$, we will need one more linear mapping
\[
g_2\colon\quad Z^2(M,\mathbb C)\to C_3(M,\mathbb C),
\]
that sends thus 2-cocycles to usual 3-chains of \emph{oriented} tetrahedra.

\begin{remark}
We use the notation~`$g_2$' to emphasize its intimate relation to morphism~$f_2$. We are not, however, going to introduce any~$g_i$ with $i\ne 2$.
\end{remark}

Let there be given a cocycle $c\in Z^2(M,\mathbb C)$ and a pentachoron~$u$ in the triangulation of~$M$, with 3-faces $t_1,\ldots,t_5 \subset u$. As $f_2(c)$ is a formal sum of unoriented tetrahedra, it has some coefficients $\beta_1,\ldots,\beta_5$ at $t_1,\ldots,t_5$. Similarly, $g_2(c)$ has some coefficients $\gamma_1,\ldots,\gamma_5$ before these tetrahedra; this time, of course, they need orientation, so we assume that it is induced from~$u$.

Mappings $f_2$ and~$g_2$ are, by definition, determined together by the following conditions.

\begin{cnd}\label{cnd:chain}
\emph{Both} following morphism compositions vanish:
\[
f_2\circ f_1=0,\qquad g_2\circ f_1=0.
\]
\end{cnd}

\begin{cnd}\label{cnd:edge}
If $c=\delta b$ for some edge~$b$, or 2-cochain $c$ just coincides with~$\delta b$ ``locally'' --- on all 2-faces of pentachoron~$u$, then \emph{only those} (three of five) $\beta_{t_i}$ can be nonzero for which $b\subset t_i$. The same applies also to~$\gamma_{t_i}$.
\end{cnd}

As locally --- to be exact, within one pentachoron --- any 2-cocycle~$c$ is the coboundary of some linear combination of edges\footnote{That is, 1-cochains taking value~$1$ on a given edge and vanishing on all other edges, see the sentence after formula~\eqref{in}.}, it follows from Condition~\ref{cnd:edge} that if we have the mentioned numbers $\beta_{t_i}$ and~$\gamma_{t_i}$ for all edges~$b$, then both $f_2$ and~$g_2$ are completely determined.

\begin{cnd}\label{cnd:sp}
The following ``scalar products'' between $\beta$'s and $\gamma$'s vanish:
\[
\beta_1\gamma_1+\dots+\beta_5\gamma_5=0
\]
for any $c$ and~$u$.
\end{cnd}

\begin{cnd}\label{cnd:5}
For any pentachoron~$u$, the rank of submatrix of~$f_2$ made of its five rows corresponding to tetrahedra $t\subset u$, equals~$5$. 
\end{cnd}

\begin{remark}\label{r:f}
We will see that these conditions determine $f_2$ and~$g_2$ up to some freedom, namely choosing signs of some square roots (see \eqref{q} and~\eqref{K}), scalings within each tetrahedron (see Subsection~\ref{ss:expl}), and overall normalization (see Convention~\ref{cnv:en}).
\end{remark}

\begin{remark}\label{r:g4f5}
It will follow also from Condition~\ref{cnd:5} that the rank of a similar submatrix of~$g_2$ (to the submatrix of~$f_2$ mentioned there) is always~$4$ (remember that we consider a \emph{generic}~$\omega$, according to Convention~\ref{cnv:g}). Condition~\ref{cnd:5} looks convenient for our construction, although other possibilities for the mentioned ranks may also be of interest.
\end{remark}

\subsection{Interpretation in terms of Grassmann--Berezin calculus: local edge operators}\label{ss:i2}

From now on, we work in the Grassmann algebra with generators~$\vartheta _t$, one for each tetrahedron~$t$ in the triangulation. Guided by our Condition~\ref{cnd:edge} and the paragraph after it (but slightly changing notations), we introduce, for a given pentachoron~$u$ and edge $b\subset u$, \emph{local edge operator} (acting on Grassmann algebra elements from the left)
\begin{equation}\label{e}
d_b=d_b^{(u)}=\sum_{\substack{t\subset u\\[.15ex] t\supset b}} (\beta_{bt} \partial_{bt} + \gamma_{bt} \vartheta _t),
\end{equation}
where we now denote as $\beta_{bt}$ and~$\gamma_{bt}$ the coefficients of tetrahedron~$t$ in $f_2(\delta b)$ and $g_2(\delta b)$, respectively.

\begin{remark}
Local edge operators were called simply `edge operators' in~\cite{full-nonlinear}. Here we want to emphasize by the word `local' that such an operator belongs to a chosen pentachoron~$u$.
\end{remark}

\begin{remark}
Local edge operators are introduced for \emph{all} edges, not only those in subset~$\mathsf B$ introduced before Lemma~\ref{l:edges}.
\end{remark}

\begin{theorem}\label{th:l}
For a given pentachoron~$u$, local edge operators span a maximal (5-dimensional) \emph{isotropic} subspace in the (10-dimensional) space of all operators of the form
  \begin{equation}\label{bpgx}
  d=\sum_{t\subset u} (\beta_t \partial_t+\gamma_t \vartheta _t),
  \end{equation}
where the scalar product is defined by means of the anticommutator, see~\eqref{scal}. This means that there exists a unique \emph{pentachoron weight}~$\mathcal W_u$ of the Grassmann--Gaussian form
\begin{equation}\label{W_u}
\mathcal W_u = \exp\left(-\frac{1}{2} \,\uptheta ^{\mathrm T} F\/ \uptheta \right),
\end{equation}
annihilated by all~$d_b$:
\begin{equation}\label{dW}
  d_b \mathcal W_u = 0.
\end{equation}
In~\eqref{W_u}, $\uptheta =\uptheta ^{(u)}$ is the column of Grassmann variables put in correspondence to 3-faces of pentachoron~$u$, for instance,
\begin{equation}\label{x}
\uptheta  = \begin{pmatrix} \vartheta _{2345} & \vartheta _{1345} & \vartheta _{1245} & \vartheta _{1235} & \vartheta _{1234} \end{pmatrix}^{\mathrm T}
\end{equation}
for $u=12345$, and $F=F^{(u)}$ is a $5\times 5$ antisymmetric matrix.
\end{theorem}

\begin{proof}
The isotropy of the space spanned by local edge operators follows from Condition~\ref{cnd:sp}: it says that any vector in that space is isotropic.

The five-dimen\-sion\-ness is guaranteed by Condition~\ref{cnd:5}.

The existence and uniqueness of Grassmann--Gaussian weight~$\mathcal W_u$~\eqref{W_u} with property~\eqref{dW} follows from Theorems \ref{th:eo} and~\ref{th:v}.
\end{proof}

Recalling that any 2-cocycle~$c$ is locally the coboundary of some linear combination of edges, \eqref{dW} implies that
\begin{equation}\label{DW}
  \left(\sum_{t\subset u} \beta_t\partial_t\right) \mathcal W_u = -\left(\sum_{t\subset u}\gamma_t \vartheta _t\right) \mathcal W_u,
\end{equation}
where $\beta_t$ and~$\gamma_t$ are coefficients of chains $f_2(c)$ and~$g_2(c)$, and tetrahedra in the right-hand side inherit the orientation from~$u$.

There are also some more elementary properties of local edge operators, and it makes sense to collect them in the following theorem.

\begin{theorem}\label{th:m}
For a given pentachoron~$u$, local edge operators are antisymmetric with respect to changing the edge orientation:
  \begin{equation}\label{da}
  d_{ij}=-d_{ji},
  \end{equation}
obey the following linear relations for each vertex $i\in u$:
  \begin{equation}\label{dv}
  \sum_{\substack{j\in u\\ j\ne i}} d_{ij} = 0,
  \end{equation}
and one more linear relation:
  \begin{equation}\label{5u}
  \sum_{b\subset u} \nu_b d_b = 0,
  \end{equation}
  where $\nu$ is any 1-cocycle such that $\omega$ makes its coboundary:
  \[
  \omega=\delta \nu,\quad \text{i.e.,}\quad \omega_{ijk}=\nu_{jk}-\nu_{ik}+\nu_{ij}.
  \]

If\/ $u$ and~$u'$ are two neighboring pentachora, and tetrahedron~$t\subset \partial u\cap \partial u'$ lies between them, then it holds for coefficients in operators $d_b^{(u)}$ and~$d_b^{(u')}$, for any edge $b\subset t$, that
  \begin{equation}\label{b+g-}
  \beta_t^{(u)}=\beta_t^{(u')},\qquad \gamma_t^{(u)}=-\gamma_t^{(u')}.
  \end{equation}
\end{theorem}

\begin{proof}
Antisymmetry~\eqref{da} is due to the fact that coefficients of local edge operators are taken from $f_2(\delta b)$ and~$g_2(\delta b)$, and $\delta b$ of course changes its sign together with~$b$.

Next, left-hand side of~\eqref{dv} corresponds to the sum of edges which is the coboundary~$\delta i$ of a single vertex~$i$ (within pentachoron~$u$). Coefficients in the l.h.s. of~\eqref{dv} come thus from $f_2(\delta^2 i)$ and~$g_2(\delta^2 i)$. So, the `one-cycle' property~\eqref{dv} follows from the fact that $\delta^2=0$.

Relation~\eqref{5u} is a consequence of Condition~\ref{cnd:chain}.

Finally, \eqref{b+g-} follows from the fact that $f_2$ deals with unoriented pentachora, while $g_2$ --- with oriented ones.
\end{proof}

\section{Partial scalar products between edge operators, and explicit formulas for them}\label{s:long}

\subsection{Partial scalar products and normalization of edge operators}\label{ss:n}

Recall that we have introduced \emph{$t$-components} and \emph{partial scalar products} for operators of the form~\eqref{bpgx}, see formulas \eqref{scs} and~\eqref{sct}. An edge operator~$d_b$ is thus the sum of its $t$-components
\begin{equation}\label{bbtgbt}
\left.d_b\right|_t = \beta_{bt}\partial_t + \gamma_{bt}\vartheta _t
\end{equation}
over tetrahedra~$t$ such that $b\subset t\subset u$, while the (vanishing) scalar product of two edge operators is a sum
\begin{equation}\label{vs}
0 = \langle d_{b_1} \mid d_{b_2} \rangle = \sum_{t\subset u} \langle d_{b_1} \mid d_{b_2} \rangle_t ,
\end{equation}
of partial scalar products, the latter being by definition the same as products $\bigl\langle d_{b_1}|_t \,\big|\, d_{b_2}|_t \bigr\rangle$ of $t$-components.

\begin{remark}
The scalar product in~\eqref{vs} vanishes, of course, due to the fact that all edge operators belong to an isotropic subspace; so, \eqref{vs} holds for differing as well as coinciding $b_1$ and~$b_2$.
\end{remark}

\begin{lemma}\label{l:s}
Choose a tetrahedron $t\subset u$ and a triangle $ijk\subset t$. Then the partial scalar product $\langle d_{ij} \mid d_{ik} \rangle_t$ remains the same under any permutation of $i,j,k$.
\end{lemma}

\begin{proof}
Take, for instance, $u=12345$, $t=1234$, and let us prove that
\begin{equation}\label{1213}
\langle d_{12} \mid d_{13}\rangle_{1234} = \langle d_{21} \mid d_{23}\rangle_{1234}.
\end{equation}
Setting $i=3$ in~\eqref{dv} and taking its $t$-component, we have (keeping in mind also~\eqref{da}):
\begin{equation}\label{d3}
-d_{13}|_{1234}-d_{23}|_{1234}+d_{34}|_{1234}=0.
\end{equation}
We want to take the scalar product of~\eqref{d3} with~$d_{12}$. As $1234$ is the only tetrahedron common for the edges $12$ and~$34$, and all edge operators are orthogonal to each other, we get
\begin{equation}\label{ortho}
\langle d_{12} \mid d_{34}\rangle_{1234}=\langle d_{12} \mid d_{34}\rangle=0.
\end{equation}
So, the mentioned scalar product, together with~\eqref{ortho}, gives~\eqref{1213} at once.
\end{proof}

\begin{lemma}\label{l:t}
For a tetrahedron $t\subset u$, construct the expression
\begin{equation}\label{s-gen}
\omega_s\langle d_{b_1} \mid d_{b_2} \rangle_t.
\end{equation}
Here tetrahedron~$t$ is considered as oriented, $s$ is any of its 2-faces with the induced orientation, and $b_1,b_2\subset s$ are two edges sharing the same initial vertex (thus also oriented). Then, the value of expression~\eqref{s-gen} does not depend on a specific choice of $s$, $b_1$ and~$b_2$, and thus belongs solely to~$t$.
\end{lemma}

\begin{proof}
We take again, to be explicit, $u=12345$, $t=1234$, and prove that
\begin{equation}\label{s}
-\omega_{123}\langle d_{12} \mid d_{13}\rangle_{1234}=\omega_{124}\langle d_{12} \mid d_{14}\rangle_{1234}
\end{equation}
(the minus sign accounts for opposite orientations of $123$ and~$124$). A small exercise shows that the following linear relation is a consequence of~\eqref{5u}:
\[
-\omega_{123}d_{13}|_{1234}-\omega_{124}d_{14}|_{1234}+\omega_{234}d_{34}|_{1234}=0.
\]
Multiplying this scalarly by~$d_{12}$ and using once again orthogonality~\eqref{ortho}, we get~\eqref{s}.
\end{proof}

\begin{lemma}\label{l:u}
Expression~\eqref{s-gen} also remains the same \emph{for all tetrahedra~$t$} forming the boundary of pentachoron~$u$, if these tetrahedra are oriented consistently (as parts of the boundary~$\partial u$).
\end{lemma}

\begin{proof}
It is enough to consider the situation where $s$ is the common 2-face of two tetrahedra $t,t'\subset u$, and show that
\begin{equation}\label{dbb}
\langle d_{b_1} \mid d_{b_2} \rangle_t = -\langle d_{b_1} \mid d_{b_2} \rangle_{t'}.
\end{equation}
Indeed, as the orientation of~$s$ as part of~$\partial t$ is different from its orientation as part of~$\partial t'$, there are two values~$\omega_s$ differing in sign, and \eqref{dbb} will yield at once that \eqref{s-gen} is the same for $t$ and~$t'$.

To prove~\eqref{dbb}, we note that $t$ and~$t'$ are the only tetrahedra containing both $b_1$ and~$b_2$, so
\[
0=\langle d_{b_1} \mid d_{b_2} \rangle = \langle d_{b_1} \mid d_{b_2} \rangle_t + \langle d_{b_1} \mid d_{b_2} \rangle_{t'}.
\]
\end{proof}

\begin{lemma}\label{l:m}
Finally, expression~\eqref{s-gen} also remains the same \emph{for the whole manifold~$M$}, if it is calculated in the following way: take any pentachoron~$u$ in the triangulation, then any tetrahedron~$t\subset u$ with the orientation induced from~$u$, then calculate~\eqref{s-gen} as in Lemma~\ref{l:t}. 
\end{lemma}

\begin{proof}
Let $t$ be the common tetrahedron for two adjacent pentachora: $t=u\cap u'$. It is enough to show that~\eqref{s-gen} is the same, regardless whether we calculate it within $u$ or~$u'$.

According to formula~\eqref{b+g-} (compare also~\cite[formulas~(58)]{full-nonlinear}),
\[
\begin{array}{rcl}
 \text{if } & d_a^{(u)}|_t=\beta_t\partial_t+\gamma_t \vartheta _t,& \text{where edge } a\subset t,\\
 \text{then   } & d_a^{(u')}|_t=\beta_t\partial_t-\gamma_t \vartheta _t.
\end{array}
\]
As a result, scalar products of edge operators change their signs on passing from $u$ to~$u'$, while the orientation of~$t$ also changes. Hence,~\eqref{s-gen} again remains the same.
\end{proof}

\begin{cnv}\label{cnv:en}
We normalize edge operators for the whole manifold~$M$ in such way that quantity~\eqref{s-gen} becomes unity:
\begin{equation}\label{s-gen=1}
\omega_s\langle d_{b_1} \mid d_{b_2} \rangle_t = 1 .
\end{equation}
\end{cnv}

Here is the matrix of scalar products $\langle d_a,d_b \rangle_t$ calculated according to Convention~\ref{cnv:en}, for $t=1234$ and~$u=12345$. The rows (resp.\ columns) correspond to edge~$a$ (resp.~$b$) taking values in lexicographic order: 12, 13, 14, 23, 24, 34:
\begin{equation}\label{sc}
\begin{pmatrix}
 \omega^{-1}_{124}{-}\omega^{-1}_{123} & \omega^{-1}_{123} & -\omega^{-1}_{124} & -\omega^{-1}_{123} & \omega^{-1}_{124} & 0 \\[1ex]
 \omega^{-1}_{123} & -\omega^{-1}_{134}{-}\omega^{-1}_{123} & \omega^{-1}_{134} & \omega^{-1}_{123} & 0 & -\omega^{-1}_{134} \\[1ex]
 -\omega^{-1}_{124} & \omega^{-1}_{134} & \omega^{-1}_{124}{-}\omega^{-1}_{134} & 0 & -\omega^{-1}_{124} & \omega^{-1}_{134} \\[1ex]
 -\omega^{-1}_{123} & \omega^{-1}_{123} & 0 & \omega^{-1}_{234}{-}\omega^{-1}_{123} & -\omega^{-1}_{234} & \omega^{-1}_{234} \\[1ex]
 \omega^{-1}_{124} & 0 & -\omega^{-1}_{124} & -\omega^{-1}_{234} & \omega^{-1}_{124}{+}\omega^{-1}_{234} & -\omega^{-1}_{234}\\[1ex]
 0 & -\omega^{-1}_{134} & \omega^{-1}_{134} & \omega^{-1}_{234} & -\omega^{-1}_{234} & \omega^{-1}_{234}{-}\omega^{-1}_{134} 
\end{pmatrix} .
\end{equation}

\begin{remark}
To calculate \emph{diagonal} elements in~\eqref{sc} is an easy exercise using linear relations similar to~\eqref{d3}.
\end{remark}

\begin{remark}\label{r:45}
Note that for tetrahedron (say)~$1235\subset 12345$, we must not only replace `4' by `5' in~\eqref{sc}, but also change all signs --- due to its different orientation! Similarly, analogues of matrix~\eqref{sc} for other tetrahedra can be calculated.
\end{remark}

\subsection[Pure differentiations and pure multiplications by~$\vartheta$'s from local edge operators]{Pure differentiations and pure multiplications by~$\boldsymbol{\vartheta}$'s from local edge operators}\label{ss:p}

\paragraph{Isotropic operators within one tetrahedron.}
We introduce, following paper~\cite{gce}, square roots of values of our 2-cocycle~$\omega$ for oriented triangles~$s=ijk$:
\begin{equation}\label{q}
q_s=\sqrt{\omega_s},
\end{equation}
fixing the signs of these roots arbitrarily.

Consider the following three linear combinations of edge operators restricted onto tetrahedron~$1234\subset 12345$:
\begin{multline}\label{p}
(q_{123}q_{124}d_{12}+q_{134}q_{234}d_{34})_{1234},\qquad
(q_{123}q_{134}d_{13}+q_{124}q_{234}d_{24})_{1234},\\
\text{and}\quad (q_{124}q_{134}d_{14}+q_{123}q_{234}d_{23})_{1234}.
\end{multline}
One can see, using the table~\eqref{sc} of scalar products, that any scalar product of two operators~\eqref{p}, including scalar squares, \emph{vanishes}. They span thus an isotropic subspace in the space of operators $\lambda\partial_{1234}+\mu \vartheta _{1234}$, and must be proportional to each other and all to either $\partial_{1234}$ or~$\vartheta _{1234}$.

The same applies if we change \emph{all} pluses in~\eqref{p} to minuses; this time we get multiples of the other one from operators $\partial_{1234}$ and~$\vartheta _{1234}$.

\paragraph{Introducing square roots of product of omegas over a tetrahedron.}
The most elegant way for further calculations seems to appear if we deal with square roots~\eqref{q} not individually, but define the following values:
\begin{equation}\label{K}
K_t = \pm \sqrt{\omega_{ijk}\omega_{ijl}\omega_{ikl}\omega_{jkl}}
\end{equation}
for every tetrahedron~$t=ijkl$. Note that the product under the square root is \emph{invariant under all permutations of vertices}~$i,\dots,l$. The signs of~$K_t$ are, by definition, subject to the following requirements:
\begin{itemize}\itemsep 0pt
 \item $K_t$ changes its sign together with the tetrahedron \emph{orientation}:
\begin{equation}\label{K_-t}
K_{-t} = -K_t,
\end{equation}
 \item for any pentachoron~$u$ with vertices $i<j<k<l<m$, if it is oriented according to this order (i.e., $u=ijklm$, see Convention~\ref{cnv:o}),
\begin{equation}\label{5K}
\prod_{t\subset u} K_t = \prod_{\substack{s=ijk\,\subset\, u\\ i<j<k}} \omega_s,
\end{equation}
where tetrahedra~$t$ in the l.h.s.\ are taken with the orientation induced from~$u$.
\end{itemize}

Thus, for instance, for pentachoron $12345$ we must have
\begin{equation}\label{5K_xmp}
K_{1234}K_{1235}K_{1245}K_{1345}K_{2345} = \omega_{123}\omega_{124}\omega_{125} \omega_{134}\omega_{135}\omega_{145} \omega_{234}\omega_{235}\omega_{245}\omega_{345},
\end{equation}
because here, in the l.h.s., exactly \emph{two} tetrahedra --- $1235$ and $1345$ --- are not oriented as induced from $12345$. Note also that the r.h.s.\ of~\eqref{5K} changes its sign together with the orientation of~$u$.

\begin{theorem}\label{th:Kq}
 \begin{enumerate}\itemsep 0pt
  \item\label{i:Kq_1} Quantities~$K_t$, having right signs satisfying \eqref{K_-t} and~\eqref{5K}, can be constructed as follows: first, order all vertices arbitrarily; second, fix arbitrarily the signs of all~$q_s$~\eqref{q} for $s=ijk$ with $i<j<k$; and third, for any tetrahedron~$t=ijkl$ with $i<j<k<l$, set
\begin{equation}\label{4q}
K_{ijkl} = q_{ijk} q_{ijl} q_{ikl} q_{jkl}.
\end{equation}
  \item\label{i:Kq_2} For a simply connected~$M$ (recall our Assumption~\ref{mpt:h}), \emph{all} such sets of~$K_t$ can be constructed according to item~\ref{i:Kq_1}, even with any one fixed ordering of vertices (thus leaving free only signs of~$q_s$).
 \end{enumerate}
\end{theorem}

\begin{proof}
 \begin{itemize}\itemsep 0pt
  \item[\ref{i:Kq_1}] As \eqref{K_-t} simply holds by definition, it remains to check~\eqref{5K}; and this is easily done directly, compare \eqref{4q} with~\eqref{5K_xmp}.
  \item[\ref{i:Kq_2}] Let there be two sets $\{K_t^{(1)}\}$ and~$\{K_t^{(2)}\}$ of quantities~$K_t$ satisfying \eqref{K_-t} and~\eqref{5K}, the first of these defined as in item~\ref{i:Kq_1} above. Then, there appears a sign $K_t^{(1)}/K_t^{(2)}=\pm 1$ for each tetrahedron~$t$. We would like to consider these signs as elements of the group~$\mathbb Z_2$ written multiplicatively, and them all as a $\mathbb Z_2$-valued \emph{3-cochain}. Moreover, this cochain is even a \emph{3-cocycle} --- because the r.h.s.\ of~\eqref{5K} is the same in both cases.

Now, according to Poincar\'e duality, the cohomology group~$H^3(M,\mathbb Z_2)$ is the same as $H_1(M,\mathbb Z_2)$, and this is trivial for a simply connected~$M$. Thus, our 3-cocycle is also a \emph{coboundary} of a certain $\mathbb Z_2$-valued 2-cochain. It can be seen now that if we change the signs of~$q_s$ for those triangles~$s$ where this cochain takes value $-1 \in \mathbb Z_2$ and then calculate the new values of~$K_t$, we will get exactly the set~$\{K_t^{(2)}\}$.
 \end{itemize}
\end{proof}

\paragraph{Superisotropic operators within one pentachoron.}
It is convenient to introduce \emph{superisotropic operators} already here, although their full meaning will be revealed only in Subsection~\ref{ss:si}, where their components will enter in the final formula~\eqref{Ftt'}. By definition, a superisotropic operator~\cite[Subsection~4.2]{full-nonlinear} is an operator of the form~\eqref{bpgx} and such that its restriction onto any tetrahedron~$t$ is isotropic. This means that, for any of five~$t\subset u$, either $\beta_t=0$ or $\gamma_t=0$.

\begin{lemma}\label{l:g}
The following linear combination of edge operators
\begin{equation}\label{su}
g = \sum_{b\subset u} \frac{\prod_{s\supset b} \omega_s}{\prod_{t\supset b} K_t} \, d_b
\end{equation}
is superisotropic. Here it is understood that the orientation of a 2-face~$s\supset b$ is always consistent with the orientation of~$b$, while the orientations of tetrahedra~$t$ are induced from~$u$.
\end{lemma}

\begin{proof}
Indeed, \eqref{su} is, when restricted to any 3-face $t\subset u$, proportional to an expression of type~\eqref{p} (here we mean of course that $K_t$'s are constructed according to Theorem~\ref{th:Kq}).
\end{proof}

\begin{remark}
The version of superisotropic operators in our previous papers~\cite{full-nonlinear,gce} is the expression~\eqref{su} multiplied by $\prod_{s\subset u} q_s$, where $q_s=\sqrt{\omega_s}$.
\end{remark}

\begin{lemma}\label{l:h}
Let $T$ be a set of either two or four tetrahedra --- 3-faces of pentachoron~$u$. If we change the signs of those~$K_t$ in~\eqref{su} for which $t\in T$, then we obtain again a superisotropic operator, and its $t$-components will be proportional to the ``old'' components if $t\notin T$, but not proportional otherwise:
\[
K_t\mapsto -K_t \;\Rightarrow\; (\textrm{multiple of\/ } \partial_t) \leftrightarrow (\textrm{multiple of\/ } \vartheta _t) \textrm{ for } t\in T.
\]
\end{lemma}

\begin{proof}
Direct calculcation using Theorem~\ref{th:Kq}. Note that changing the signs of an even number of~$K_t$'s is compatible with~\eqref{5K}.
\end{proof}

We must find now possible identifications between the components of superisotropic operators, on the one side, and multiples of either~$\partial_t$ or~$\vartheta _t$, on the other, in such way as to comply with Condition~\ref{cnd:5}. This latter can be reformulated, in terms of our five-dimensional space of isotropic operators, as the requirement that this space contains operators of the form
\begin{equation}\label{one_d}
\partial_t+(\text{multiples of }\vartheta _{t'}\text{ for the four tetrahedra }t'\ne t)
\end{equation}
for all five~$t\subset u$.
This means that one possibility is to identify the components of~$g$, given by \eqref{su}, with multiples of~$\partial_t$. Indeed, we can then, due to Lemma~\ref{l:h}, change an even number of~$\partial_{t'}$, namely those where $t'\ne t$, to~$\vartheta _{t'}$, and obtain operators proportional to~\eqref{one_d}.

Condition~\ref{cnd:5} holds as well if we choose an \emph{even} number of components of $g$ given by \eqref{su}, and identify them with multiples of corresponding~$\vartheta _t$ (while others still with multiples of~$\partial_t$). This is clear due to Lemma~\ref{l:h}.

But if we did such identification with an \emph{odd} number of components of~$g$, then we would be able to construct --- again due to Lemma~\ref{l:h} --- a superisotropic operator with all components proportional to~$\vartheta _t$'s, and this cannot exist in the same isotropic space with the five operators~\eqref{one_d}.

\begin{remark}
A small exercise shows also that the rank of submatrix of~$g_2$ mentioned in Remark~\ref{r:g4f5} cannot be~$5$: otherwise, we would be able to construct a superisotropic operator with all components proportional to~$\vartheta _t$'s.
\end{remark}

\paragraph{Global identification of operators~$\boldsymbol g$ components with differentiations.}
Although each tetrahedron~$t$ belongs to \emph{two} different pentachora in the manifold triangulation, the $t$-components of two corresponding operators~$g$ can easily be shown to be proportional. For instance, for $t=1234$, and $K_t$'s as in Theorem~\ref{th:Kq}, they are proportional to all expressions~\eqref{p}, and these can at most change their sign on passing to the other $u\supset t$.

As we have seen, Condition~\ref{cnd:5} implies that an even number of components of operator~$g$~\eqref{su}, for any triangulation pentachoron, \emph{must} be identified with multiples of~$\vartheta _t$'s. Due to Lemma~\ref{l:h}, we can change the signs of square roots~$K_t$ so as to have only  multiples of~$\partial_t$'s for all components. We have thus proved the following theorem.

\begin{theorem}\label{th:a}
Assuming Condition~\ref{cnd:5}, the signs of square roots~$K_t$~\eqref{K} (satisfying \eqref{K_-t} and~\eqref{5K}) can always be chosen so that all $t$-components of all operators~$g$, defined according to \eqref{su}, give multiples of~$\partial_t$.
\qed
\end{theorem}

Conversely, explicit formulas in the next Subsection~\ref{ss:expl} show that edge operators can be constructed in such way that the components of~$g$'s are proportional to differentiations, and this for \emph{any} permitted choice of $K_t$ signs. Recall that all such choices, adopting Assumption~\ref{mpt:h}, are described in Theorem~\ref{th:Kq}.

\subsection{Explicit formulas for edge operators}\label{ss:expl}

Here are some properties that local edge operators \emph{must} obey:
\begin{itemize}\itemsep 0pt
 \item relations \eqref{da} and~\eqref{dv},
 \item linear combinations of their components of type~\eqref{p} --- note that these can be easily re-written in terms of~$K_t$ --- must be multiples of pure differentiations,
 \item while the same combinations with pluses replaced with minuses must be multiples of multiplications by~$\vartheta _t$.
\end{itemize}
These requirements give rise, within each tetrahedron~$t$, to a system of linear homogeneous equations on coefficients~$\beta_t$ and~$\gamma_t$ for six operators~$d_b$,\, $b\subset t$, and this is enough to determine, without contradictions, the six~$\beta_t$ up to an overall factor, and the six~$\gamma_t$ up to another overall factor.\footnote{This is checked using computer algebra; compare Convention~\ref{cnv:g}.} One of these factors can then be expressed through the other due to the normalization~\eqref{s-gen=1}, but the other one remains free to change (not taking value zero). One such free `scaling', or `gauge', factor arises within each tetrahedron, and for any choice of them, we obtain good operators $f_2$ and~$g_2$ satisfying all our defining Conditions \ref{cnd:chain}--\ref{cnd:5} and Convention~\ref{cnv:en}.

Recall that we denote the coefficients~$\beta_t$ and~$\gamma_t$ for a given operator~$d_b$,\, $b\subset t$, as $\beta_{bt}$ and~$\gamma_{bt}$, see formulas \eqref{e} and~\eqref{bbtgbt}. It is desirable to have these coefficients expressed uniformly, that is, by the same formula(s) for all six edges $b\subset t$. Computer algebra allows doing this even in many different ways (which implies also differents gauges), and the simplest expressions that could be obtained look as follows.

First, define the following auxiliary quantities for a tetrahedron~$t=ijkl$:
\[
\psi_{ij,t} = \psi_{ji,t} = \omega_{ijk}\omega_{ijl}+\omega_{ikl}\omega_{jkl} - \omega_{ijk}^2-\omega_{ijl}^2.
\]
Coefficients~$\beta_{bt}$ have the following form:
\begin{equation}\label{betaijt}
\beta_{ij,t}=-\beta_{ji,t}= -(\omega_{ijk}+\omega_{ijl})\omega_{ikl}\omega_{jkl}\psi_{kl} - (\omega_{ikl}+\omega_{jkl})\psi_{ij}K_t.
\end{equation}
Note that changing the tetrahedron orientation can be written as $k\leftrightarrow l$ together with $K_t \mapsto -K_t$, hence $\beta_{ij,t}$ does not depend on the tetrahedron orientation, as required.

Gammas are then like betas, but with the changed sign of~$K_t$, and an additional factor that we call~$1/c_t$:
\begin{equation}\label{gammaijt}
\gamma_{ij,t} = -\gamma_{ji,t} = \frac{1}{c_t} \bigl( -(\omega_{ijk}+\omega_{ijl})\omega_{ikl}\omega_{jkl}\psi_{kl} + (\omega_{ikl}+\omega_{jkl})\psi_{ij}K_t \bigr).
\end{equation}
Here
\begin{multline*}
c_t = -\omega_{ijk}\,\omega_{ijl}\,\omega_{ikl}\,\omega_{jkl}\, (\omega_{ijl}-\omega_{ijk
 } )\, (\omega_{ikl}+\omega_{ijk} )
 \, (\omega_{ikl}-\omega_{ijl} )\\
\cdot \bigl( 24\,\omega_{ijk}\,\omega_{ijl}\,\omega_{ikl}\,\omega_{jkl}+\frac{1}{2} (\omega_{ijk}^2 + \omega_{ijl}^2 + \omega_{ikl}^2 + \omega_{jkl}^2)^2 \bigr).
\end{multline*}
Note that $c_t$, and hence $\gamma_{ij,t}$ as well, does change its sign together with the tetrahedron orientation, again as required.

Having now at hand explicit expressions for all~$\beta_{bt}$ in terms of the cocycle~$\omega$, we have thus determined matrix~$f_2$, see the sentence after Condition~\ref{cnd:edge}. As for the gammas, we will need them soon as well.

\section[Morphism~$f_3$, and justifying the name ``chain complex'']{Morphism~$\boldsymbol{f_3}$, and justifying the name ``chain complex''}\label{s:f3}

\subsection[Definition of~$f_3$ and proof of chain property]{Definition of~$\boldsymbol{f_3}$ and proof of chain property}\label{ss:f3}

Denote~$\Uptheta$ the column made of all Grassmann variables~$\vartheta _t$ that we have put in correspondence to each tetrahedron~$t$ in the triangulation:
\[
\Uptheta=\begin{pmatrix} \vartheta _1 & \ldots & \vartheta _{N_3} \end{pmatrix}^{\mathrm T}
\]
(recall our notational Convention~\ref{cnv:n}). Antisymmetric matrix~$f_3$ is, by definition, such that the quadratic form $\Uptheta^{\mathrm T} f_3 \, \Uptheta$ is the sum of all quadratic forms entering in weights~\eqref{W_u}:
\[
\Uptheta^{\mathrm T} f_3 \, \Uptheta = -\frac{1}{2} \sum_{\substack{\mathrm{all}\\ \!\!\!\!\mathrm{pentachora\; }u\!\!\!\!}} \uptheta ^{\mathrm T} F\/ \uptheta .
\]

\begin{theorem}\label{th:c}
Sequence~\eqref{c-symb} is indeed a chain complex.
\end{theorem}

\begin{proof}
Clearly, it is enough to prove that
\[
f_2\circ f_1=0\qquad\text{and}\qquad f_3\circ f_2=0.
\]

The first of these equalities holds, however, already by construction, see Condition~\ref{cnd:chain}, and can, of course, be checked using explicit formulas for matrix elements of~$f_2$ in Subsection~\ref{ss:expl}, and description of morphism~$f_1$ in the end of Section~\ref{s:t}.

It remains to show that $f_3\circ f_2=0$. Let $c$ be a 2-cochain. Then $f_2(c)$ is a chain of unoriented tetrahedra, denote its coefficient at a tetrahedron~$t$ as~$\beta_t$. Consider the purely differential operator
\[
D_c=\sum_t \beta_t \partial_t.
\]
We want to show that
\[
D_c \exp(\Uptheta^{\mathrm T} f_3 \, \Uptheta) = D_c \prod_u \mathcal W_u = 0.
\]
Indeed, it easily follows from combining formula~\eqref{DW}, Leibniz rule~\eqref{Leibniz}, and the fact that each~$\gamma_t$ appears two times and with different signs (because each tetrahedron belongs to two pentachora, and with different induced orientations).

On the other hand,
\[
D_c \exp(\Uptheta^{\mathrm T} f_3 \, \Uptheta) = -2 \, \Uptheta^{\mathrm T} f_3 \begin{pmatrix} \beta_1 \\ \vdots \\ \beta_{N_3} \end{pmatrix} \exp(\Uptheta^{\mathrm T} f_3 \, \Uptheta),
\]
which means that
\[
f_3 \begin{pmatrix} \beta_1 \\ \vdots \\ \beta_{N_3} \end{pmatrix} = \begin{pmatrix} 0 \\ \vdots \\ 0 \end{pmatrix}
\]
and, consequently, $f_3\circ f_2=0$.
\end{proof}

\subsection[Superisotropic operators and explicit formulas for matrix~$F$ entries]{Superisotropic operators and explicit formulas for matrix~$\boldsymbol{F}$ entries}\label{ss:si}

The rows of matrix~$F$ entering in the pentachoron weight~\eqref{W_u} correspond to superisotropic operators in the following sense: every component of the column
\begin{equation}\label{px}
\mathsf p+F\uptheta 
\end{equation}
is superisotropic. In~\eqref{px}, $\mathsf p$ and~$\uptheta $ are columns corresponding to five differentiations~$\partial_t$ and five multiplications by Grassmann generators~$\vartheta _t$, \ $t\subset u$. For instance, if $u=12345$, then $\uptheta $ is given by~\eqref{x} and $\mathsf p$, similarly, by
\[
\mathsf p = \begin{pmatrix} \partial_{2345} & \partial_{1345} & \partial_{1245} & \partial_{1235} & \partial_{1234} \end{pmatrix}^{\mathrm T}.
\]

We know from Theorem~\ref{th:a} and Subsection~\ref{ss:expl} (recall the paragraph after Theorem~\ref{th:a}) that the superisotropic operator containing only differentiations can, and must, be identified, up to a scalar factor, with operator~$g$~\eqref{su}, with some permitted choice of (signs of) square roots~$K_t$. Moreover, Lemma~\ref{l:h} tells us that the operator proportional to the $i$-th entry in~\eqref{px}, and thus corresponding to the tetrahedron~$t$ \emph{not containing} the vertex~$i$, is obtained from~\eqref{su} by the following change of signs:
\begin{equation}\label{K:znaki}
K_{t'} \mapsto -K_{t'} \quad \text{for} \quad t' \ne t.
\end{equation}
We thus identify the entries in column~\eqref{px}, up to scalar factors, with the operators given by~\eqref{su}, but with the signs changed according to~\eqref{K:znaki}. We denote such operators~$g^{(t)}$, and their coefficients~$\alpha_b^{(t)}$:
\begin{equation}\label{gtdb}
g^{(t)} = \sum_{b\subset u} \alpha_b^{(t)} d_b,
\end{equation}
where $\alpha_b^{(t)}$ is obtained from coefficients in~\eqref{su} by substitution~\eqref{K:znaki}:
\[
\alpha_b^{(t)} = \prod_{s\supset b} \omega_s \Big/ \prod_{t'\supset b} (\epsilon_{tt'} K_{t'}), \qquad \epsilon_{tt'} = \begin{cases} 1 & \text{if \ } t=t', \\ -1 & \text{if \ } t\ne t'. \end{cases}
\]

As we said, operator~$g^{(t)}$ is a linear combination of one differentiation~$\partial_t$ and four multiplications by~$\vartheta _{t'}$, \ $t'\ne t$, and we denote the corresponding coefficients as $\beta_t$ and~$\gamma_{tt'}$:
\begin{equation}\label{gt}
g^{(t)} = \beta_t\partial_t + \sum_{\substack{t'\subset u\\ t'\ne t}} \gamma_{tt'}\vartheta _{t'}.
\end{equation}
Comparing \eqref{gt} with \eqref{gtdb} and~\eqref{bbtgbt}, we see then that
\[
\beta_t = \sum_{b\subset t} \alpha_b^{(t)} \beta_{bt},\qquad \gamma_{tt'} = \sum_{b\subset t'} \alpha_b^{(t)} \gamma_{bt'}.
\]

\begin{remark}
Recall that edge operators~$d_b$ involve only tetrahedra containing edge~$b$; this is the reason why the sums above are taken over edges $b\subset t$ or~$t'$.
\end{remark}

Finally, since operators~\eqref{gt} are proportional to those in the column~\eqref{px}, and these latter have differentiations with coefficients~$1$, here is the formula for matrix~$F$ elements:
\begin{equation}\label{Ftt'}
F_{tt'} = \frac{\gamma_{tt'}}{\beta_t}.
\end{equation}

\section{Discussion}\label{s:d}

The present work develops the constructions proposed in~\cite{KS2,full-nonlinear,gce}. In this Part~I, we have constructed an exotic chain complex for one fixed triangulation. In Part~II~\cite{part-II}, we will investigate what happens under Pachner moves, and construct an invariant of the pair ``manifold~$M$, element in~$H^2(M,\mathbb C)$''. In part~III, we plan to present calculations for different specific manifolds.

\subsubsection*{Acknowledgements}

Felix Berezin taught me that a `fermionic' problem often turns out to be simpler than its `bosonic' counterpart. I would like to pay homage to him and his work that made possible this work of mine.

I would like also to thank the creators and maintainers of Maxima\footnote{\href{http://maxima.sourceforge.net/}{http://maxima.sourceforge.net/}} computer algebra system for their great work.

\end{document}